\documentclass[a4paper,11pt,dvipdfm]{article}
\usepackage{theorem}
\usepackage{amsmath,amssymb,amsfonts}
\usepackage{graphicx}

\renewcommand{\title}[1]{\vspace{\fill}
\eject\addtolength{\baselineskip}{4pt}
{\bfseries\LARGE #1}\\[3mm]\addtolength{\baselineskip}{-4pt}}
\renewcommand{\author}[3]{\parbox[t]{75mm}
{\begin{center}{\scshape #1}\\[3mm] #2\\
 {\ttfamily #3} \end{center}}}
 
\theorembodyfont{\rmfamily}

\newtheorem{thm}{\bfseries Theorem}
\newtheorem{lem}[thm]{\bfseries Lemma}        
\newtheorem{prop}[thm]{\bfseries Proposition} 
\newtheorem{cor}[thm]{\bfseries Corollary}     
\newtheorem{defn}[thm]{\bfseries Definition}

\newenvironment{proof}{\medskip                    
\noindent{\scshape Proof:}}{\quad $\Box$\medskip}  

\begin{document}
\begin{center}
\title{Families of polytopal digraphs that do not satisfy \\ the shelling property}

\author{David Avis\footnotemark[1]
}{
Graduate School of Informatics\\ Kyoto University \\ and \\
School of Computer Science and GERAD \\ McGill University
}{
avis@cs.mcgill.ca
}\footnotetext[1]{Research is supported by NSERC, FQRNT and KAKENHI.}
\author{
Hiroyuki Miyata\footnotemark[2]}
{
Graduate School of Information Sciences,
Tohoku University, \\
Aoba-ku, Sendai, Miyagi, Japan.
}{
hmiyata@dais.is.tohoku.ac.jp
}
\author{
Sonoko Moriyama\footnotemark[2]
}{ 
Graduate School of Information Sciences,
Tohoku University, \\
Aoba-ku, Sendai, Miyagi, Japan. 
}{
moriso@dais.is.tohoku.ac.jp
}\footnotetext[2]{Research is supported by KAKENHI.}

\end{center}

\begin{center}
\vspace{3mm}\today\vspace{3mm}
\end{center}

\begin{quote}
{\bfseries Abstract:}
A polytopal digraph $G(P)$ is an orientation of the skeleton of a
convex polytope $P$.
The possible non-degenerate pivot operations of the simplex method
in solving a linear program over $P$ can be represented as a special
polytopal digraph known as an LP digraph.
Presently there is no general characterization of which polytopal
digraphs are LP digraphs, although four necessary properties are known:
acyclicity, unique sink orientation(USO), the Holt-Klee property
and the shelling property.
The shelling property was introduced by Avis and Moriyama (2009),
where two examples are given in $d=4$ dimensions of polytopal 
digraphs  satisfying
the first three properties but not the shelling property.
The smaller of these examples has $n=7$ vertices.
Avis, Miyata and Moriyama(2009) constructed for each $d \ge 4$ and $n \ge d+2$,
a $d$-polytope $P$ with $n$ vertices which has a polytopal digraph 
which is an acyclic USO that satisfies the Holt-Klee property,
but does not satisfy the shelling property.
The construction was based on a minimal such example, which has $d=4$ and $n=6$.
In this paper we explore the shelling condition further.
First we give an apparently stronger definition of the shelling property, 
which we then prove is 
equivalent to the original definition. Using this stronger
condition we are able to give a more general construction of such families.
In particular, we show that given any
4-dimensional polytope $P$ with $n_0$ vertices whose unique sink is simple, 
we can extend $P$ for any $d \ge 4$ and $n \ge n_0 + d-4$ to a $d$-polytope with these
properties that has $n$ vertices.
Finally we investigate the strength of the shelling condition
for $d$-crosspolytopes, 
for which Develin (2004) has given a complete characterization of LP orientations.
\end{quote}

\begin{quote}
{\bf Keywords: polytopal digraphs, simplex method, shellability, polytopes }
\end{quote}
\vspace{5mm}

\section{Introduction} 

Let $P$ be a $d$-dimensional convex polytope ($d$-polytope) in $\Re^d$.
We assume that the reader is familiar with polytopes,
a standard reference being \cite{Ziegler}.
The vertices and extremal
edges of $P$ form an (abstract) undirected graph called the skeleton of $P$.
A {\it polytopal digraph} $G(P)$ is formed
by orienting each edge of the skeleton of $P$ in some manner.
In the paper, when we refer to a polytopal digraph $G(P)$ we shall
mean the pair of both the digraph and the polytope $P$ itself,
not just the abstract digraph.

We can distinguish four properties that the digraph $G(P)$
may have, each of which has been well studied:
\begin{itemize}
\item 
{\it Acyclicity}: 
$G(P)$ has no directed cycles.
\item 
{\it Unique sink orientation (USO)}~ (Szab{\'{o}} and Welzl~\cite{SW01}):
Each subdigraph of $G(P)$ induced by a non-empty face of $P$ has a unique source
and a unique sink.
\item 
{\it Holt-Klee property}~(Holt and Klee~\cite{HK}):
$G(P)$ has a unique sink orientation, and
for every $k (\geq 0)$-dimensional face ($k$-face) $H$ of $P$ there are 
$k$ vertex disjoint paths 
from
the unique source to the unique sink of $H$ in the subdigraph
$G(P,H)$ of $G(P)$ induced by $H$.
\item 
{\it LP digraph}:
There is a linear function $f$
and a polytope $P^{'}$ combinatorially equivalent to $P$
such that for each pair of vertices $u$ and $v$ of $P^{'}$ that form
a directed edge $(u,v)$ in $G(P^{'})$, we have
$f(u) < f(v)$. 
(LP digraphs are called polytopal digraphs in Mihalsin and Klee \cite{MK00}.)
\end{itemize}

Interest in polytopal digraphs stems from the fact that the simplex method with
a given pivot rule can be 
viewed as an algorithm for finding a path to the sink in a polytopal digraph
which is an LP digraph.
Research continues on pivot rules for the simplex method
since they leave open the possibility of finding a strongly polynomial time algorithm
for linear programming. 
An understanding of which polytopal digraphs are LP digraphs is therefore
of interest.
The other three properties are necessary
properties for $G(P)$ to be an LP-digraph.
We note here that Williamson Hoke~\cite{Ho88}
has defined a property called \emph{complete unimodality}
which is equivalent to a combination of
acyclicity and unique sink orientation.
LP digraphs
are completely characterized when $d = 2,3$.
Their necessary and sufficient properties in dimension $d = 2$ are
a combination of acyclicity and unique sink orientation,
i.e., complete unimodality,
and those in $d = 3$ are
a combination of acyclicity, unique sink orientation
and the Holt-Klee property~\cite{MK00}.
On the other hand,
no such characterization is known yet for higher dimensions.

Another necessary property 
for $G(P)$ to be an LP digraph
is based on \emph{shelling},
which is one of the fundamental tools of polytope theory.
A formal definition of shelling is given in Section 2.
Let $G(P)$ be a polytopal digraph 
for which the polytope $P$ has $n$ vertices labelled $v_1, v_2, ... , v_n$.
A permutation $r$ of the vertices is a \emph{topological sort} of $G(P)$
if, whenever $(v_i , v_j)$ is a directed edge of $G(P)$, $v_i$ precedes $v_j$
in the permutation $r$.
Let $L(P)$ be the face lattice of $P$.
A polytope $P^*$ whose face lattice is $L(P)$ "turned upside-down" is called
a \emph{combinatorially polar} polytope of $P$.
Combinatorial polarity interchanges vertices of $P$ with facets of $P^*$.
We denote by $r^*$ the facet ordering of $P^*$ corresponding to the vertex ordering 
of $P$ given by $r$.

\begin{itemize}
\item{\it Shelling property}~(Avis and Moriyama~\cite{AM08}):
There exists a topological sort $r$ of $G(P)$ such
that the facets of $P^*$ ordered by $r^*$ are a shelling of $P^*$.
\end{itemize}

Results relating these properties of polytopal digraphs
have been obtained by various authors. 
Figure~\ref{fig:relationship_general} summarizes the relationship among the 4 properties,
where the regions $F,G,H,I,J,X,Y$ are non-empty. For further information, see \cite{AM08}.
\begin{figure}[ht]
\begin{center}
\includegraphics[scale=0.30]{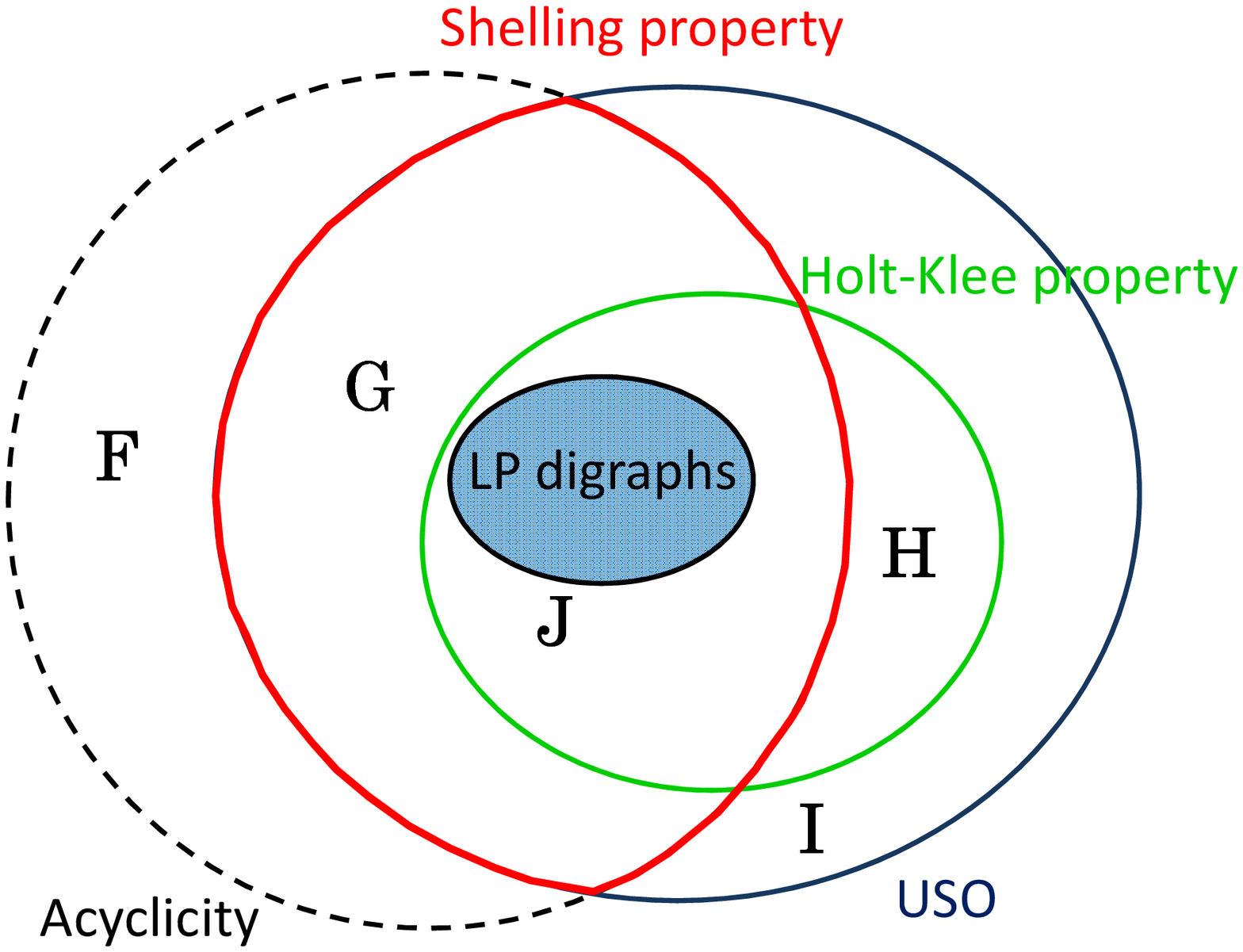}
\includegraphics[scale=0.30]{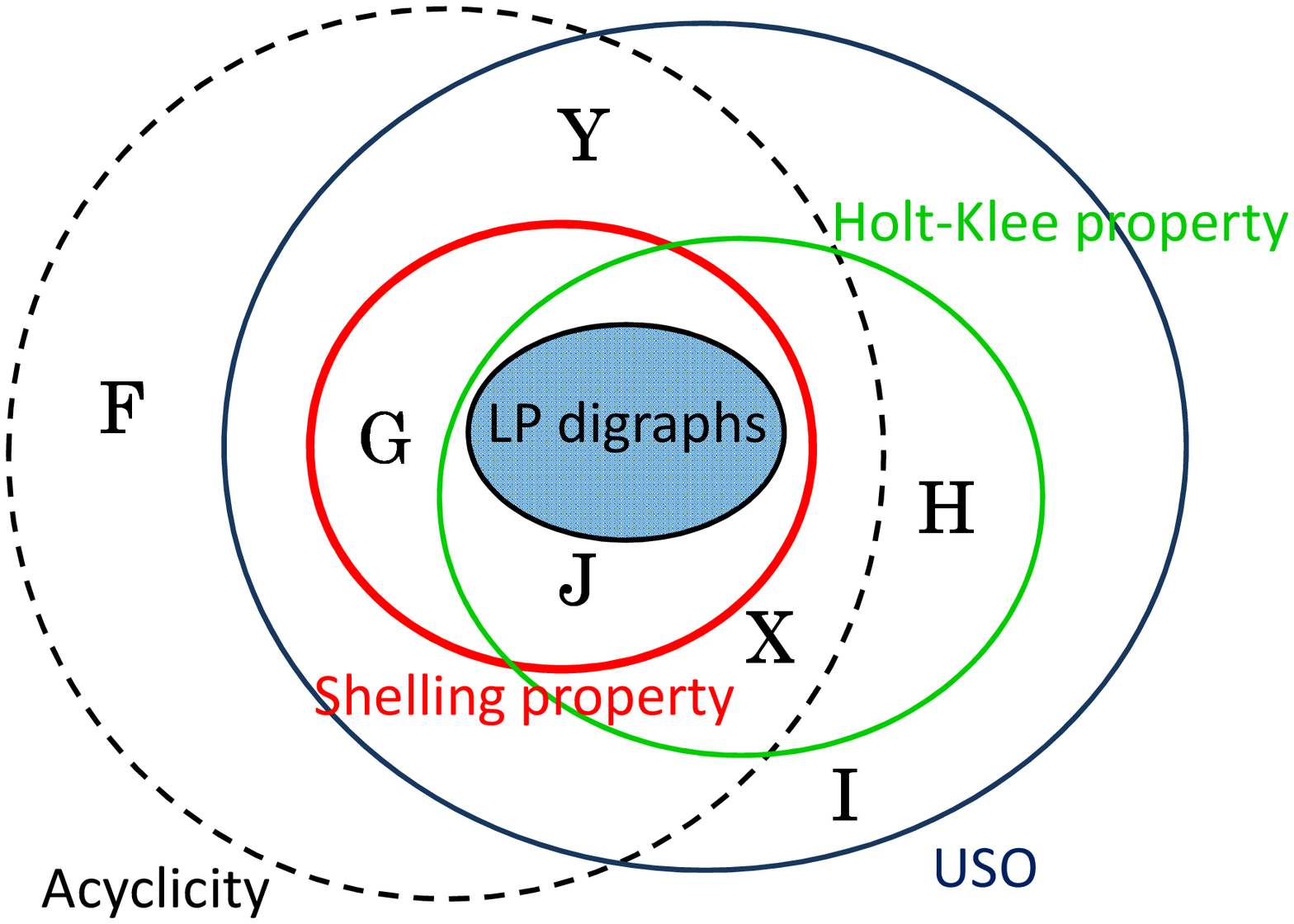}
\end{center}
\caption{The Relationships among the 4 properties when $P$ is simple[left] and general[right]}
\label{fig:relationship_general}
\end{figure}
\\
When we restrict ourselves to the case $d \leq 3$, the region $X$ is empty
since the shelling property is characterized by acyclic USOs for $d \leq 3$~\cite{AM08}.
The region $X$ is also empty for simple polytopes because the shelling property is equivalent to the acyclic USO property
for simple polytopes \cite{AM08,Ho88}.
On the other hand, two examples of non-simple polytopal digraphs in the region $X$ had been known for $d=4$:
Develin's example~\cite{Mi06} is a polytopal digraph
on the skeleton of a $4$-dimensional crosspolytope with eight vertices,
and the example proposed by Avis and Moriyama~\cite{AM08}
is a polytopal digraph a $4$-polytope with seven vertices.

The existence of the non-empty region $X$ shows the
importance of the shelling property, namely that
there exist polytopal digraphs
satisfying the three existing necessary properties
for LP digraphs, but not the shelling property.
This motivates the following definition.
\begin{defn}
A polytopal graph $G(P)$ is called 
an {\it $X$-type graph} if it is an acyclic USO satisfying the Holt-Klee property,
but not the shelling property.
\end{defn}

In our paper in proceedings of 6th Japanese-Hungarian Symposium on Discrete Mathematics and
Its Applications\cite{AMM09}, we proved the following result.
\begin{thm}\label{thm:main_2}
For every $d \geq 4$ and every $n \geq d+2$,
there exists a $d$-polytope $P$ with $n$ vertices
which has an $X$-type  graph.
\end{thm}

In this paper, we improve the result in the proceedings version as follows.
First we reconsider the definition of the shelling property and give a new definition, which
looks stronger than the original definition.
\begin{defn}\label{defn:main}~\mbox{{\rm (a new definition of the shelling property)}}\\
$G(P)$ is acyclic, and 
the facets of $P^*$ ordered by $r^*$ are a shelling of $P^*$ for any topological sort $r$ of $G(P)$.
\end{defn}
Actually, this property will be proved to be equivalent to the original shelling property.
Using the equivalence of the original definition and the new one, 
we give general constructions of families of $X$-type graphs, which generalize constructions in \cite{AMM09}.
In particular we prove the following:
\begin{thm}\label{thm:main_3}
Let $P_0$ be a 4-polytope with $n_0$ vertices which have an $X$-type graph $G(P_0)$
whose unique sink is simple.
Then 
$P_0$ can be extended to
a $d$-polytope $P$ with $n$ vertices
having an $X$-type graph for every $d \geq 4$ and every $n \geq n_0+d-4$.
\end{thm}

Next, we investigate orientations of $d$-crosspolytope, motivated by results of Develin~\cite{Mi06}.
He introduced a very useful notion called pair sequences, by which we can characterize LP orientations and acyclic USOs
of $d$-crosspolytope. In particular, from his results, we see that
LP orientations of $d$-crosspolytope can completely be characterized by the shelling property.
The relationship among the $4$ properties of LP digraphs is summarized as follows for $d$-crosspolytope ($d \geq 4$),
where the regions $K,L,M$ are non-empty. 
\begin{figure}[ht]
\begin{center}
\includegraphics[scale=0.30,clip]{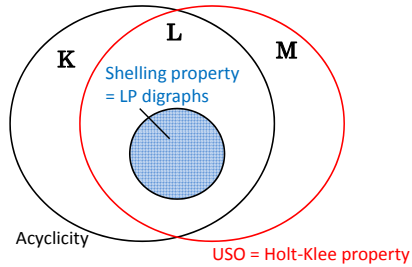}
\end{center}
\caption{The relationship among the 4 properties for $d$-crosspolytope}
\label{fig:crosspoltope_relationship}
\end{figure}

Using a characterization of LP orientations of $d$-crosspolytope by Develin~\cite{Mi06}, we can estimate quantitatively
how strong the shelling property is, compared with the other 3 properties in the case of $d$-crosspolytope (Section 6).

\section{Preliminaries}
We use the definition of shelling given in Ziegler \cite[Definition 8.1]{Ziegler}
which is slightly more restrictive than
the one used by Brugesser and Mani \cite{ BM71}.
Let $P$ be a $d$-polytope in $\Re^d$.
A \emph{shelling} of $P$ is a linear ordering $F_1, F_2, \cdots, F_s$
of the facets of $P$ such that
either the facets are points,
or it satisfies the following conditions \cite[Definition 8.1]{Ziegler}:
\begin{quote}
(i) the first facet $F_1$ has a linear ordering of its facets which is a shelling of $F_1$.\\
(ii) For $1 < j \leq m$
the intersection of the facet $F_j$ with the previous facets
is non-empty and is a beginning segment of a shelling of $F_j$, that is,
$$
F_j \cap \bigcup_{i {=} 1}^{j-1}{F_i}
= G_1 \cup G_2 \cup \cdots \cup G_r
$$
for some shelling $G_1, G_2, \cdots, G_r, \cdots, G_t$ of $F_j$,
and $1 \leq r \leq t$.
(In particular this requires that
{\em all maximal faces included in $F_j \cap \bigcup_{i {=} 1}^{j-1}{F_i}$
have the same dimension $d-2$}.)
\end{quote}
Any polytope has at least one shelling
because of the existence of \emph{line shellings}~\cite{BM71}, described below.
Hence the condition (i) is in fact redundant \cite[Remark 8.3 (i)]{Ziegler}.


Let $P$ be a $d$-polytope with $m$ facets in $\Re^d$.
A directed straight line $L$ 
that intersects the interior of $P$ and
the affine hulls of the facets of $P$ at distinct points is called
\emph{generic} with respect to $P$.
We choose a generic line $L$ and label a point
interior to $P$ on $L$ as $x$. 
Starting at $x$, we number consecutively the intersection points along $L$ with
the affine hulls of facets as $x_1, x_2 , ... , x_m$, wrapping around at infinity,
as in Figure \ref{fig:lineShelling}.
The ordering of the corresponding facets of $P$
is the \emph{line shelling} of $P$ generated by $L$.
Every line shelling is a shelling of $P$ (see, e.g., \cite{Ziegler}).
\begin{figure}[ht]
\begin{center}
\includegraphics[scale=0.40,clip]{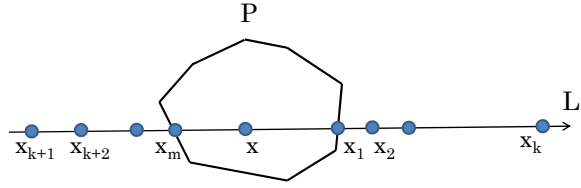}
\end{center}
\caption{The intersection points along a directed straight line $L$}
\label{fig:lineShelling}
\end{figure}
\\
\textbf{Notation:}

In this paper we will sometimes identify a face of $P$ with the set of extreme points of $P$ that it contains
when it does not cause confusion.
The following binary operators will be used when we use such a representation.
Let $f_1,f_2$ be finite sets representing a face of a polytope $P$ and that of a polytope $Q$ respectively.
We will write $f_1 \equiv f_2$ if $f_1$ and $f_2$ coincide as sets
and in addition have same face structures.
That is, 
\begin{align*}
f_1 \equiv f_2 \Leftrightarrow
&\text{ $\exists v_1,..,v_n$ s.t. $f_1=f_2=\{v_1,...,v_n \}$ and } \\
&\text{ $\max \{ f_1 \cap F \mid F$ is a facet of $P \} = \max \{ f_2 \cap G \mid G$ is a facet of $Q \}$, } \\
&\text{where $\max S$ denotes the set of all maximal subsets of $S$ for a set $S$.}
\end{align*}
We will use symbols $\cap, \cup$
to denote the intersection of facets and the union of facets respectively.
On the other hand, we will use $\Cap$ and $\Cup$ to denote
the intersection operator and union operator for finite sets.

\section{Examples of 4-polytopes with X-type graphs}

We showed in Theorem 2 of \cite{AMM09}
an example of a polytope $\Lambda$ in dimension $4$ and with $6$ vertices
which admits an X-type graph $G(\Lambda)$,
and constructed a single family of X-type graphs
based on $G(\Lambda)$ in Theorem 3 of \cite{AMM09}.
The unique sink of $G(\Lambda)$ 
is a simple vertex of the $4$-polytope $\Lambda$.
In this paper we give a more general construction,
which can be applied to
an X-type graph of any $4$-polytope $P_0$ 
with the unique sink which is a simple vertex of $P_0$.
Each new X-type graph therefore extends to a new family.
In this section, we introduce one such X-type graph,
$G(\Omega)$, as shown in Figure \ref{fig:4-polytope_v8f10}.
\begin{figure}[ht]
\begin{center}
\includegraphics[scale=0.35]{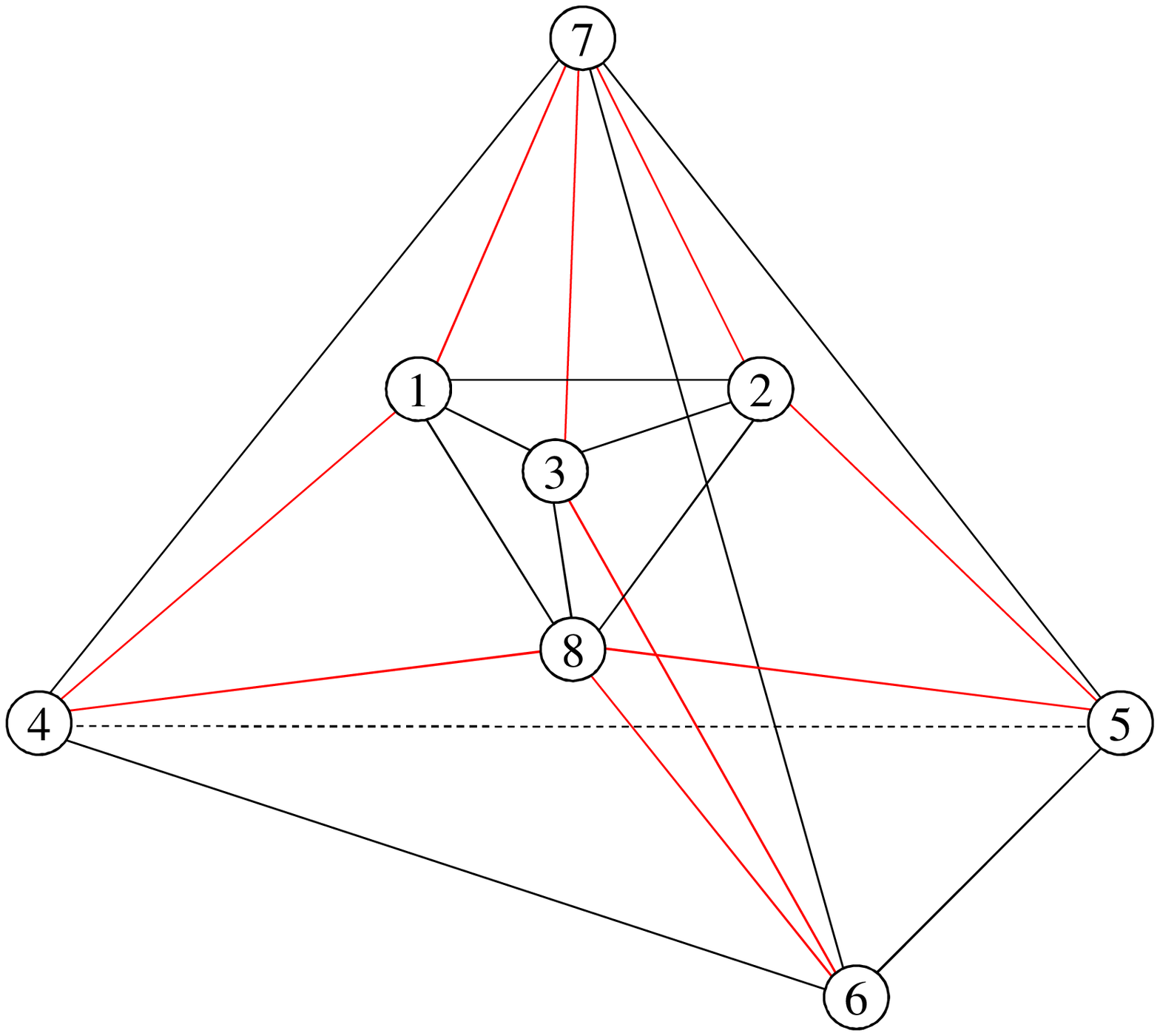}
\includegraphics[scale=0.35]{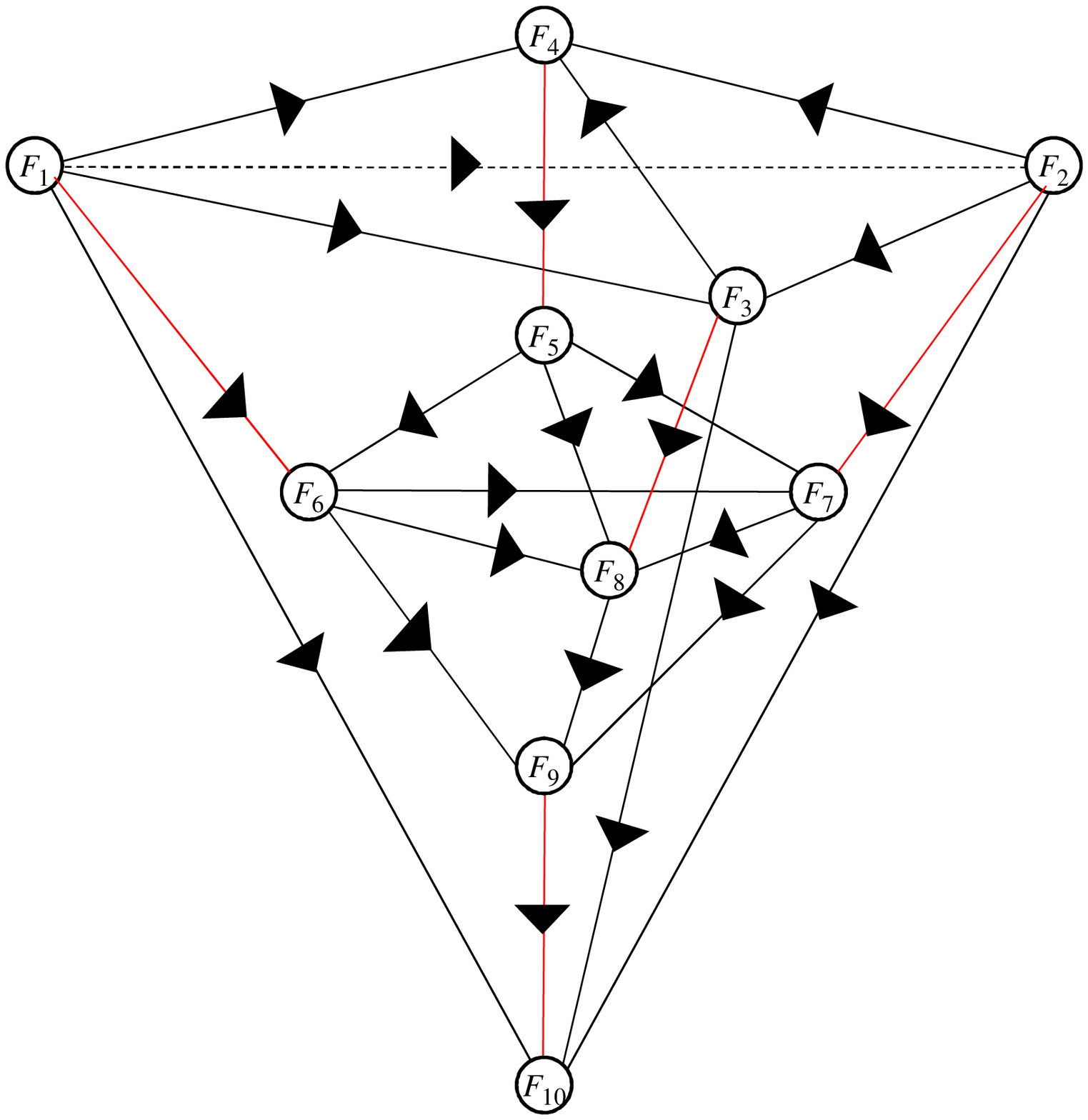}
\end{center}
\caption{The combinatorial polar $\Omega^*$ of a $4$-polytope $\Omega$
with $X$-type graph $G(\Omega)$}
\label{fig:4-polytope_v8f10}
\end{figure}

Let $\Omega^*$ be the $4$-dimensional polytope
with eight vertices $1,2,...,8$ and ten facets 
facets $F_1, F_2, ..., F_{10}$
shown in Figure \ref{fig:4-polytope_v8f10}.
The coordinates of the eight vertices
are given in Table \ref{table:coordinates_v8f10}
and the supporting hyperplanes of the ten facets of $\Omega^*$ are given as
$a_1 x_1 + a_2 x_2 + a_3 x_3 + a_4 x_4 \leq b$ with the coefficients
in Table \ref{table:Facets_counterExample_v8f10}.
\begin{table}
\begin{center}
{\scriptsize
\begin{tabular}{|| c| c@{}c@{}c@{}c || c| c@{}c@{}c@{}c
|| c| c@{}c@{}c@{}c ||  c| c@{}c@{}c@{}c || } \hline
vertex & $( x_1,$ & $x_2,$ & $x_3,$ & $x_4 )$ & vertex & $( x_1,$ & $x_2,$ & $x_3,$ & $x_4 )$
	& vertex & $( x_1,$ & $x_2,$ & $x_3,$ & $x_4 )$ & vertex & $( x_1,$ & $x_2,$ & $x_3,$ & $x_4 )$ \\ \hline \hline
  $1$ &  $( -2,$ & $ 1,$ & $ 1,$ & $  1 )$
& $2$ &  $( 2,$ & $ 1,$ & $  1,$ & $  1 )$
& $3$ &  $( 0,$ & $  -2,$ & $ 1,$ & $ 1 )$
& $4$ &  $( -4,$ & $ 2,$ & $  -2,$ & $  -1 )$ \\ \hline
  $5$ &  $(  4,$ & $  2,$ & $ -2,$ & $  -1 )$
& $6$ &  $( 0,$ & $ -4,$ & $ -2,$ & $  -1 )$
& $7$ &  $( 0,$ & $ 0,$ & $ 2,$ & $  -1 )$
& $8$ &  $( 0,$ & $ 0,$ & $ -1,$ & $ 1 )$ \\ \hline
\end{tabular}
}
\caption{The coordinates of the eight vertices
of the $4$-dimensional polytope $\Omega^*$}
\label{table:coordinates_v8f10}
\end{center}
\end{table}

\begin{table}
\begin{center}
{\scriptsize
\begin{tabular}{|| c | c | c@{}c@{}c@{}c@{}c || c | c | c@{}c@{}c@{}c@{}c ||} \hline
facet & vertices  & $( a_1,$ & $ a_2,$ & $ a_3,$ & $a_4,$ & $ b )$
& facet & vertices & $( a_1,$ & $ a_2,$ & $ a_3,$ & $a_4,$ & $ b )$ \\ \hline \hline
  $F_{1}$ & 1, 2, 4, 5, 7                
	& $( 0,$ & $ 4,$ & $ 2,$ & $-1,$ & $ 5 )$ 
& $F_{2}$ & 1, 3, 4, 6, 7                
	& $( -3,$ & $ -2,$ & $ 2,$ & $-1,$ & $ 5 )$  \\ \hline 
  $F_{3}$ & 2, 3, 5, 6, 7                
	& $( 3,$ & $ -2,$ & $ 2,$ & $ -1,$ & $ 5 )$ 
& $F_{4}$ & 1, 2, 3, 7                   
	& $( 0,$ & $ 0,$ & $ 2,$ & $ 1,$ & $ 3 )$  \\ \hline 
  $F_{5}$ & 1, 2, 3, 8                   
	& $( 0,$ & $ 0,$ & $ 0,$ & $ 1,$ & $ 1 )$
& $F_{6}$ & 1, 2, 4, 5, 8                
	& $( 0,$ & $ 4,$ & $-2,$ & $ 5,$ & $ 7 )$ \\ \hline
  $F_{7}$ & 1, 3, 4, 6, 8                
	& $(-3,$ & $-2,$ & $-2,$ & $ 5,$ & $ 7 )$ 
& $F_{8}$ & 2, 3, 5, 6, 8                
	& $( 3,$ & $-2,$ & $-2,$ & $ 5,$ & $ 7 )$  \\ \hline 
  $F_{9}$ & 4, 5, 6, 8                   
	& $( 0,$ & $ 0,$ & $-2,$ & $ 1,$ & $ 3 )$ 
& $F_{10}$ & 4, 5, 6, 7                  
	& $( 0,$ & $ 0,$ & $ 0,$ & $-1,$ & $ 1 )$  \\ \hline 
\end{tabular}
}
\caption{The ten facets of the $4$-dimensional polytope $\Omega^*$
in Figure \ref{fig:4-polytope_v8f10}
and the coefficients of their supporting hyperplanes
for the coordinates of the eight vertices in Table \ref{table:coordinates_v8f10}}
\label{table:Facets_counterExample_v8f10}
\end{center}
\end{table}

The correctness of this $V$ and $H$-representation can be checked
by standard software such as $cdd$~\cite{cdd}, $lrs$~\cite{lrs} or PORTA~\cite{porta}.
Let  $\Omega$ be a combinatorial polar of $\Omega^*$.
A polytopal digraph $G(\Omega)$ on $\Omega$ is also shown in
Figure \ref{fig:4-polytope_v8f10}.
We have the following result.
\begin{prop} \label{prop:X-typeGraph_newExample}
$G(\Omega)$ is an acyclic USO that satisfies the Holt-Klee property,
but not the shelling property.
\end{prop}
\begin{proof}
In $G(\Omega)$ all edges are directed from smaller index to larger index,
so it satisfies acyclicity.
The facets of the polytope $\Omega$
consist of six triangular prisms and two hexahedra.
It is easy to check that each of these facets and their proper faces
has a unique source and sink
by referring to  Figure \ref{fig:4-polytope_v8f10}.
In the case of dimension two, any acyclic USO digraph satisfies the Holt-Klee property.
There exist three vertex-disjoint paths from the unique source to the unique sink
in each facet, and four vertex disjoint paths
from $F_1$ to $F_{10}$ in $G(\Omega)$
(see Figure \ref{fig:4-polytope_v8f10}).
It follows that the graph $G(\Omega)$ is
an acyclic USO that satisfies the Holt-Klee property.
We prove that the graph $G(\Omega)$
does not satisfy the shelling property.

The following proof is the same as the proof \cite[Proposition 4]{AMM09},
which shows that $G(\Lambda)$ does not satisfy the shelling property.
There is a path $F_1, F_2, ..., F_{10}$
of all the ten vertices of $G(\Omega)$
in order of their indices, and so this ordering
is the unique topological sort of the graph.
By referring to Table \ref{table:Facets_counterExample_v8f10}
and Figure \ref{fig:4-polytope_v8f10}, 
we see that the first three facets of $\Omega^*$ in this order are:
$F_1 = \{ 1, 2, 4, 5, 7 \}$,
$F_2 = \{ 1, 3, 4, 6, 7 \}$, and
$F_3 = \{ 2, 3, 5, 6, 7 \}$.
Therefore $F_3 \cap \bigcup_{i=1}^{2}{F_i}$
is the union of two 2-faces of $\Omega^*$: 
one with vertices $\{ 2, 5, 7 \}$ and one with vertices $\{ 3,6,7 \}$.
These two faces intersect at a single vertex
and so cannot be the beginning of a shelling of $F_3$.
Hence the unique topological sort of $G(\Omega)$
is not a shelling of $\Omega^*$.
This completes the proof.
\end{proof}

We remark that $G(\Omega)$ is not any of the $X$-type graphs $G(\Lambda_{4,n})$ for $n \geq 6$,
which are generated from $G(\Lambda)$ using truncation operations \cite[Proposition 7]{AMM09}.
The polytope $\Lambda_{4,n}$ has $n$ vertices and $n+1$ facets \cite[Lemma 6]{AMM09}.
while $\Omega$ has ten vertices and eight facets.
It follows that $G(\Omega)$ generates another family of $X$-type graphs
different from those based on $G(\Lambda_{4,n})$ for $n \geq 6$.

\section{Equivalence of the definitions of the shelling property}
In this section, we prove that the new definition of the shelling property (Definition~\ref{defn:main}) is equivalent to the original
definition and hence LP digraphs satisfy this condition.
To prove the equivalence, we first prove the following lemma.
\begin{lem}
Let $P$ be a $d$-polytope with vertices $v_1,...,v_n$ and $G(P)$ an acyclic USO of $P$. 
For each vetex $v$ of $P$, let us denote $F(v)$ to describe the facet of $P^*$ corresponding to $v$(, which can be represented by
the set of all facets of $P$ that contain $v$).
Then
\[ F(v_k) \cap \bigcup_{i<k}{F(v_i)} = \bigcup_{j: (v_j, v_k) \in G(P)}{F(v_k) \cap F(v_j)} \]
for $k \leq n$, where $(v_j,v_k)$ denotes the edge directed from $v_j$ to $v_k$.
\end{lem}
\begin{proof} 
($\supset$) is trivial and we prove ($\subset$) part.

Let $f$ be a facet of $P$ with $F(f) \in F(v_k) \cap \bigcup_{i<k}{F(v_i)}$.
Then 
\[ F(f) \in F(v_k) \cap F(v_{i_0})\]
for some $i_0 < k$, where $F(f)$ is the vertex of $P^*$ corresponding to $f$.
This means 
\[ v_k,v_{i_0} \in f. \]
We assume now that there is no $j \leq k$ s.t. ($v_j, v_k$) is an edge of $G(f)$.
In this case, $v_k$ is a unique source of $f$.
This contradicts to $v_{i_0} \in f$ ($i_0 < k$).
Therefore we can take $j < k$ such that $(v_k, v_j)$ is an edge of $G(f)$ and
\[ F(f) \in \bigcup_{j: (v_j,v_k) \in G(P)}{F(v_k) \cap F(v_j)}. \] 
\end{proof}

Next we prove the equivalence of the two definitions of the shelling property.
\begin{thm}
Let $P$ be a $d$-polytope and $G(P)$ a polytopal digraph of $P$.
Suppose that there is a topological sort $v_1,...,v_n$ of $G(P)$ s.t.
$F(v_1),...,F(v_n)$ is a shelling of $P^*$. 
Then for any topological sort $w_1,...,w_n$ of $G(P)$,
$F(w_1),...,F(w_n)$ is a shelling of $P^*$.
\end{thm}
\begin{proof}
By lemma, we have
\[ F(w_i) \cap \bigcup_{j=1}^{i-1}{F(w_j)} = \bigcup_{j: (w_j, w_i) \in G(P)}{F(w_i) \cap F(w_j)}\]
Let $v_k = w_i$. Then the following holds:
\[ \bigcup_{j: (w_j, w_i) \in G(P)}{F(w_i) \cap F(w_j)} = \bigcup_{j: (v_j, v_k) \in G(P)}{F(v_k) \cap F(v_j)}. \]
Since $F(v_1),...,F(v_n)$ is a shelling of $P^*$, the right hand side of the above relation is a beginning segment of some shelling of $F(v_k)=F(w_i)$.
This implies that the left hand side is also a beginning segment of some shelling of $F(w_i)$.
\end{proof}

The equivalence makes it easier to check whether a given polytopal digraph satisfies the shelling property or not, that is,
it suffices to consider one specific topological sort.
\begin{cor}
Let $P$ be a $d$-polytope and $G(P)$ an LP-digraph on $P$. 
For any topological sort $v_1,...,v_n$ of $G(P)$, $F(v_1),...,F(v_n)$ is a shelling of $P^*$.
\end{cor}

\section{Construction of an infinite family of $X$-type graphs}
In this section, we give a proof of Theorem \ref{thm:main_3}.
It is based on two operations. The first, truncation, builds
a $4$-polytope with $n+1$ vertices from one with $n$ vertices.
The second, pyramid, increases the dimension of a polytope and the number
of its vertices by one.
These operations are useful to construct infinite families of polytopes with certain properties.
They are also used in \cite{FMO09} to construct an infinite family of {\it non-HK* oriented matroids}.
Truncation operation considered in \cite{FMO09} is for $3$-polytopes, but we consider truncation operation
for $4$-polytopes, here.
\subsection{Truncation operation}
In this section we describe the first operation, which is
the \emph{truncation} of a $4$-polytope.
\begin{defn}(Truncation of 4-polytopes)\label{df:truncation}
Let $P$ be a $4$-polytope in $\Re^4$
containing a simple vertex $v$,
i.e., a vertex $v$ with exactly four neighbours,
$\{ v_i~:~i = 1,2,\dots ,4 \}$ the vertices adjacent to $v$,
and $\{ u_j~:~ j = 1,2 \}$ points
in the relative interior of  edge $(v, v_j)$.
The {\it truncated polytope} $\mbox{tr} (P, v)$
is a $4$-polytope $P \cap (H \cup H^+)$
where $H$ is the hyperplane determined by $u_1$, $u_2$, $v_3$, $v_4$ 
and $H^+$ is the open half-space of $H$ containing all vertices except $v$,
see Figure \ref{fig:truncation}.
\end{defn}
\begin{figure}[ht]
\begin{center}
\includegraphics[scale=0.35]{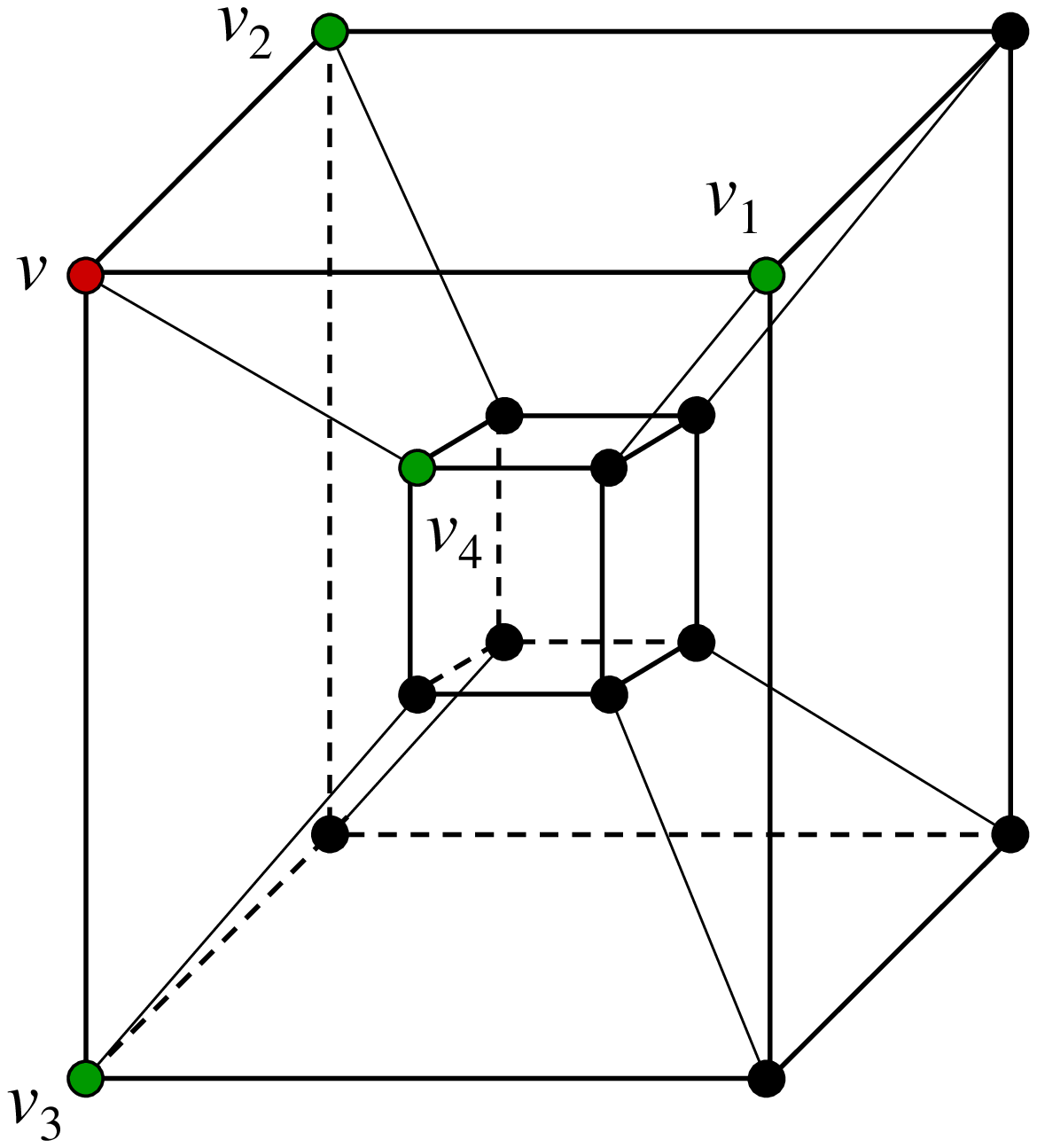}
\includegraphics[scale=0.35]{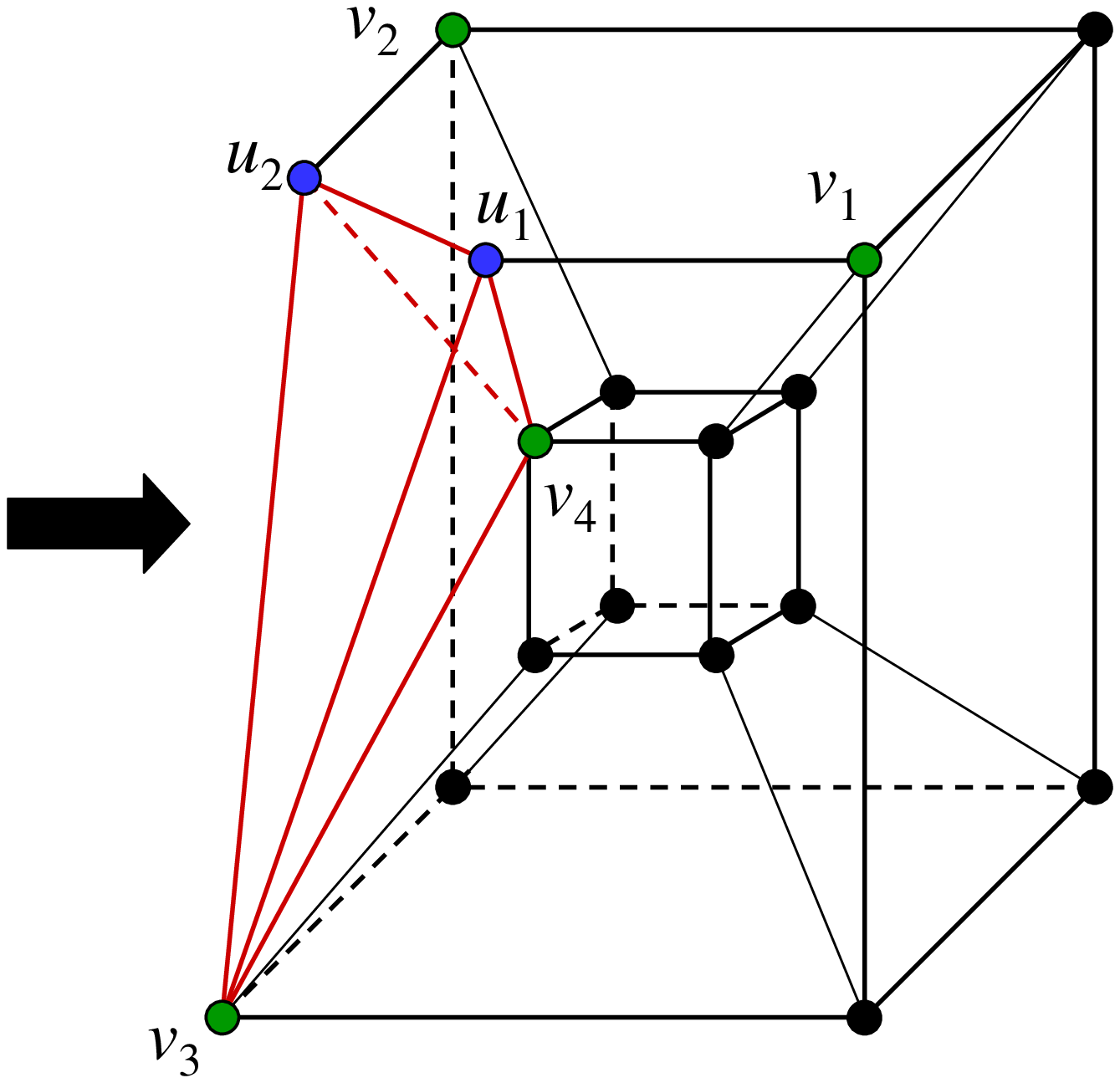}
\end{center}
\caption{A truncation operation for a $4$-polytope}
\label{fig:truncation}
\end{figure}
We define truncation operation for acyclic polytopal digraphs as follows.
\begin{defn}(Truncation of acyclic polytopal digraphs)
Let $G(P)$ be an acyclic polytopal digraph of a $4$-polytope $P$ with a simple vertex $v$
that is also a unique sink.
We define the polytopal digraph $tr(G(P))$ as follows.
We assume without loss of generality that there is no directed path from $v_4$ to $v_3$ in $G(P)$
and define the following orientation of the graph of the truncated polytope $tr(P,v)$.
We do not change the orientations of edges existing in $G(P)$ and orient new edges as follows:
\[ (v_3,v_4) \ (\text{if it is not an edge of $G(P)$}),\]
\[ (v_1,u_1),(v_3,u_1),(v_4,u_1),(v_2,u_2),(v_3,u_2),(v_4,u_2),(u_1,u_2). \] 
\end{defn}
\begin{figure}[ht]
\begin{center}
\includegraphics[scale=0.37]{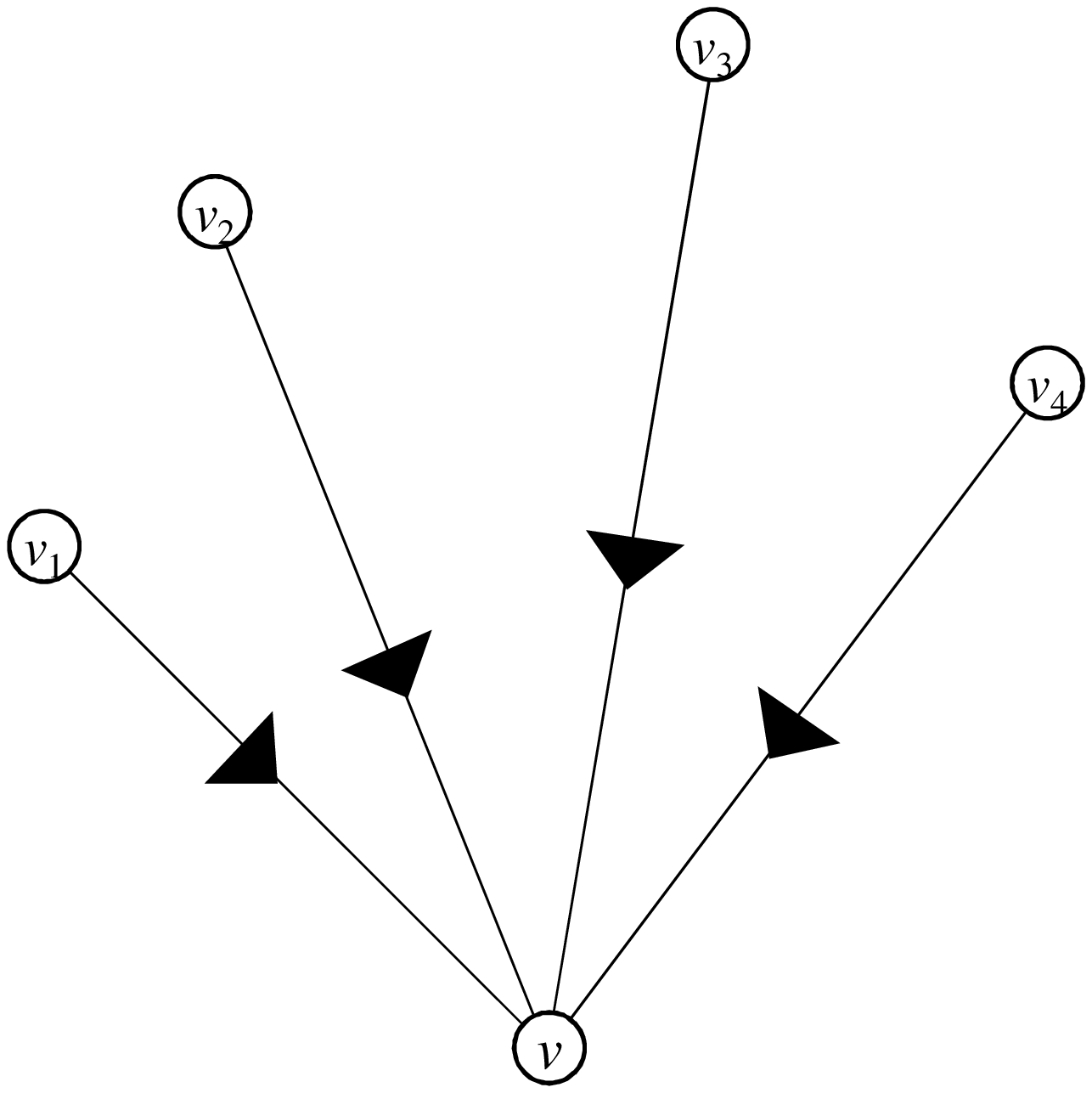}
\includegraphics[scale=0.37]{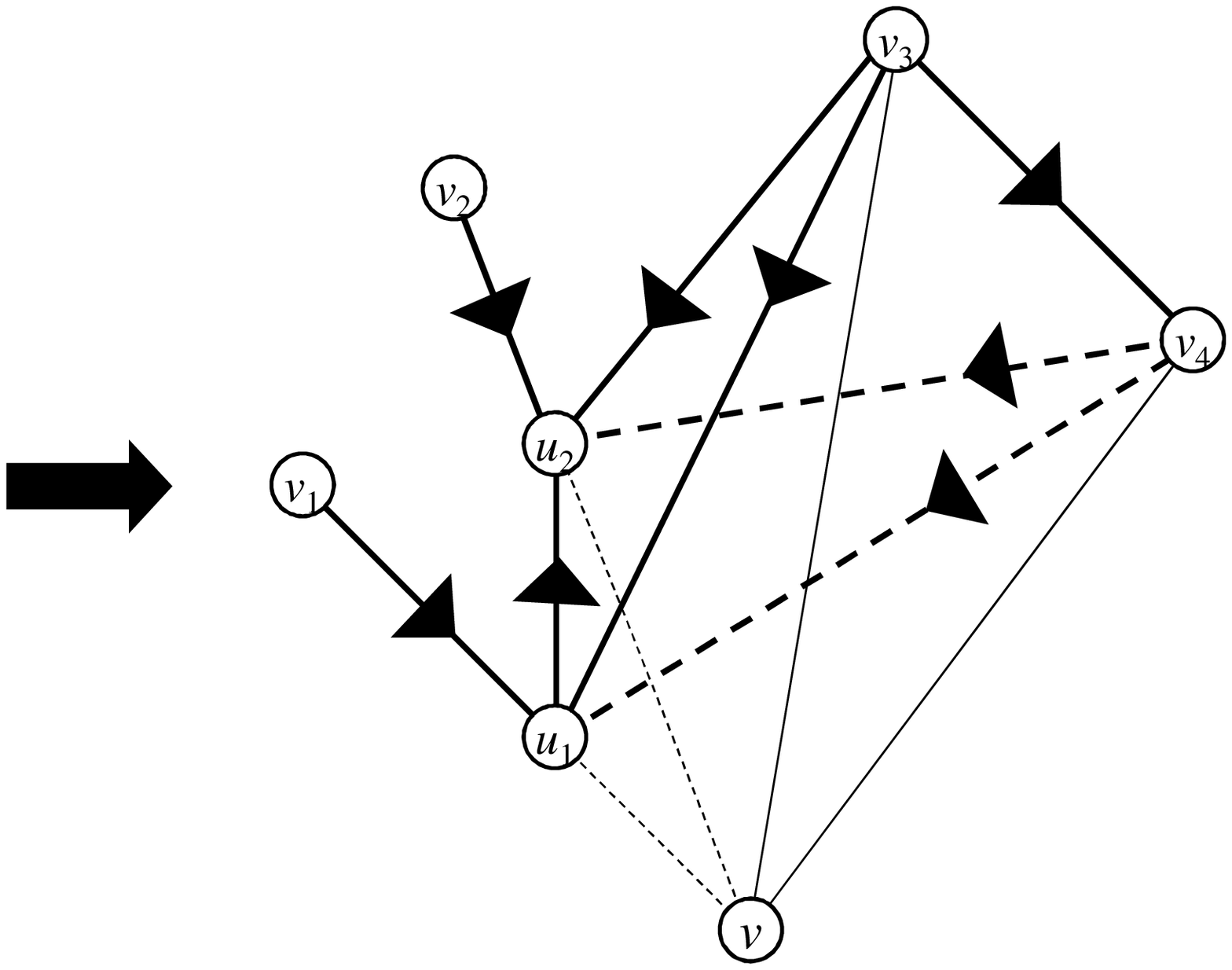}
\end{center}
\caption{A truncation operation for a digraph of a $4$-polytope}
\label{fig:graph_truncation}
\end{figure}

\begin{lem}
Given a $4$-polytope $P$ which has a simple vertex $v$,
let $F_{1},...,F_{m}$ be facets of $P$ not containing $v$,
and let $F_{m+1},...,F_{m+4}$ the four facets of $P$ containing $v$ labelled as follows:
$F_{m+1} \supset \{ v_{2},v_{3},v_{4},v \} $,
$F_{m+2} \supset \{ v_{1},v_{3},v_{4},v \} $,
$F_{m+3} \supset \{ v_{1},v_{2},v_{4},v \} $ and
$F_{m+4} \supset \{ v_{1},v_{2},v_{3},v \} $.

Then all of the facets of $tr(P)$ can be listed as follows.
\[ F'_{1}:=F_{1},\dots ,F'_{m}:=F_{m}, \
F'_{m+1}:= G_{m+1} \Cup \{ u_{2},v_2, v_{3},v_{4} \} ,\
F'_{m+2}:= G_{m+2}  \Cup \{ u_{1},v_1, v_{3},v_{4} \} ,\] 
\[ F'_{m+3}:= G_{m+3}  \Cup \{ u_{1},u_{2},v_1,v_2,v_{4} \} ,\
F'_{m+4}:= G_{m+4} \Cup \{ u_{1},u_{2},v_1,v_2,v_{3} \} ,\
F'_{m+5}:= \{ u_{1},u_{2},v_{3},v_{4} \} , \]
where
$G_{m+1} := F_{m+1} \setminus \{ v_{2},v_{3},v_{4},v \} $,
$G_{m+2} := F_{m+2} \setminus \{ v_{1},v_{3},v_{4},v \} $,
$G_{m+3} := F_{m+3} \setminus \{ v_{1},v_{2},v_{4},v \} $ and
$G_{m+4} := F_{m+4} \setminus \{ v_{1},v_{2},v_{3},v \} $.
\end{lem}
\begin{proof}
Facets of polytopes are determined by supporting hyperplanes.
Every supporting hyperplane of $tr(P)$ is one of the following types.
\begin{itemize}
\item It is also a supporting hyperplane of a facet of $P$. 
\item It is not a supporting hyperplane of any facet of $P$. In this case, this hyperplane is
the hyperplane $H$ in Definition \ref{df:truncation}.
\end{itemize}
Note that all supporting hyperplanes of $P$ are also supporting hyperplanes of $tr(P)$.
All the facets $F$ contained in $H^+$ are preserved by truncation since $F \cap (H \cup H^+)=F$.
On the other hand, $F_{m+1},F_{m+2},F_{m+3},F_{m+4}$ are changed.
Consider the supporting hyperplane $H_{m+1}$ of $F_{m+1}$.
The corresponding facet of $tr(P)$ is computed as follows.
\[ H_{m+1} \cap tr(P) = H_{m+1} \cap (H \cup H^+) \cap P = G_{m+1} \Cup \{ u_{2},v_2, v_{3},v_{4} \}.\]
Similarly, new facets $F'_{m+2},F'_{m+3},F'_{m+4}$ are obtained as in the statement of this lemma.
Finally, note that
the new facet $F'_{m+5}=\{ u_{1},u_{2},v_{3},v_{4} \}$ is obtained as the facet corresponding to the supporting hyperplane $H$. 
\end{proof}

\begin{lem}\label{lem:line}
Given a $d$-polytope $Q$ in $\Re^{d}$, let $H_{1},H_{2}$ be facets of $Q$
having common ($d-2$)-face $H_{1} \cap H_{2}$ of $Q$.
Then $Q$ has a shelling beginning with $H_{1},H_{2}$.
Similarly, $Q$ has a shelling ending with $H_{2},H_{1}$.
\end{lem}
\begin{proof}
Let ${\bf p},{\bf q},{\bf q'} \in \Re^{d}$
be relative interior points of $H_1 \cap H_2$, $H_1$, $H_2$ respectively
and consider a point
${\bf r} = \frac{{\bf p}+\epsilon_1 {\bf q}}{1+\epsilon_1 } \in H_1$
and
${\bf r'} = \frac{{\bf p}+\epsilon_2 {\bf q'}}{1+\epsilon_2 } \in H_2$
for sufficiently small $\epsilon_1 > 0$ and $\epsilon_2 > 0$.
Then the line shelling with respect to a generic line $L:
\{ ({\bf p}-{\bf r}+{\bf p}-{\bf r'} )t+{\bf r} \mid t \in \Re^{d} \}$,
perturbing $q$ and $q'$ if necessary, satisfies the condition. 
\end{proof}

We prove the following theorem using the lemma.
\begin{thm} \label{prop:1stConstruction}
Let $G(P)$ be a polytopal digraph of a polytope $P$ satisfying the conditions in Theorem \ref{thm:main_3}.
If $G(P)$ is an $X$-type graph, $tr(G(P))$ is also an $X$-type graph.
\end{thm}
\begin{proof}
In this proof,
we use notation $F(w)$ to describe the facet of $P^*$ corresponding to
a vertex $w$ of $P$ and $F'(w')$ a vertex $w'$ of $tr(P,v)^{*}$.

First, we notice that $tr(G(P))$ is acyclic
since no edge is directed from $u_1,u_2$ to
the subgraph $G(P)_{P \setminus \{ v \} }$ induced by $P \setminus \{ v \}$ of $G(P)$,
and both $G(P)_{P \setminus \{ v \} }$ and the subgraph induced by $u_1, u_2$ are acyclic.

Then we consider the USO property and the Holt-Klee property.
Let $H$ be an $i$-face of $tr(P)$.
Since $H$ satisfies the properties if $i < 2$ clearly, we assume $i \geq 2$.
We proceed by case analysis.\\
{\it {\bf Case 1:}} $H$ contains both $u_1$ and $u_2$.

In this case,
$H$ is $F'_{m+5}$ itself or is contained in one of the facets $F'_{m+1}, \cdots , F'_{m+4}$,
which we denote by $F'$.
If $H$ is identical to $F'_{m+5}$,
it satisfies the USO property and the Holt-Klee property clearly,
thus we consider the other cases.
We denote the corresponding facet in $P$ by $F$,
and consider a vertex $p \neq u_1,u_2$.
Then there is a path $p \rightarrow \dots \rightarrow w \rightarrow v$ in $G(P)_{F}$,
where $w$ is one of the vertices $v_1, \dots ,v_4$.
Therefore we can take a path
$p \rightarrow \dots \rightarrow w \rightarrow u_1 \rightarrow u_2$ in $tr(G(P))_{F'}$,
and thus $p$ is not a sink.
We can prove the uniqueness of a source similarly.
In addition, the above discussion also confirms the Holt-Klee property.\\
{\it {\bf Case 2:}} $H$ contains $u_1$ and does not contain $u_2$.

We can also prove the USO property and the Holt-Klee property similarly using a path
$p \rightarrow \dots \rightarrow w \rightarrow u_1$ instead.\\
{\it {\bf Case 3:}} $H$ contains $u_2$ and does not contain $u_1$.

We use a path $p \rightarrow \dots \rightarrow w \rightarrow u_2$ instead.\\
{\it {\bf Case 4:}} $H$ contains neither $u_1$ nor $u_2$.\\
(Case 4-a) $H$ is contained in one of the facets $F'_{1},\dots ,F'_{m}$.\\
Since $F'_1 \equiv F_1,\dots ,F'_m \equiv F_m$, $H$ 
clearly satisfies the USO property and the Holt-Klee property.
\\
(Case 4-b) $H$ is contained in none of the facets $F'_{1},\dots ,F'_{m}$.\\
Since $dim(H) \geq 2$, $H$ should be identical to $F'_{m+1} \cap F'_{m+2}=(G_{m+1} \Cap G_{m+2}) \Cup \{ v_3,v_4 \}$,
where $G_{m+1} \Cap G_{m+2} \neq \phi$.
In this case, $v_4$ is a unique sink and there are two directed paths from the source $s$ of $tr(G(P))_{H}$ to the sink $v_4$:
\[ s \rightarrow p \rightarrow \dots \rightarrow w \rightarrow v_4, \ s \rightarrow p' \rightarrow \dots \rightarrow w' \rightarrow v_3 \rightarrow v_4, \]
where $s \rightarrow p \rightarrow \dots \rightarrow w \rightarrow v_4 \rightarrow v$ and 
$s \rightarrow p' \rightarrow \dots \rightarrow w' \rightarrow v_3 \rightarrow v$ are vertex disjoint paths in $G(P)_{H}.$

From the above four cases,
we conclude that $tr(G(P))$ satisfies the USO property and the Holt-Klee property.

Next, we consider the shelling property.
In order to prove that if $G(P)$ does not satisfy the shelling property, $tr(G(P))$ does not also satisfy the shelling property, 
we prove the contra-positive. 
First, we can take a topological sort $w_1,...,w_{l-2},u_1,u_2$ of $tr(G(P),v)$, where $w_1,..,w_{l-2} \in G(P) \setminus \{ v \}$. Then 
$F(w_1),...,F(w_{l-2})$, $F(u_1),F(u_2)$ is a shelling of $tr(P,v)$ by the new definition of the shelling property.
We prove that $F(w_{1}),\dots ,F(w_{l-2}),F(v)$ is a shelling of $P$ using Lemma~\ref{lem:line}
as follows.

For $w \notin \{ v_{3},v_{4},v \}$, we have $w' \notin \{ v_{1},v_{2},v_{3},v_{4},v \} $ for $w' < w$ and
\[ F(w) \cap \bigcup_{w' \in ver(P), w' < w}{F(w')} =  F'(w') \cap \bigcup_{w' \in ver(P), w' < w}{F'(w')}, \]
is the beginning of a shelling of $F(w) \equiv F'(w)$.

We consider the remaining parts. First, we notice
\[ F(v_{3}) \cap F(w') \neq \phi , w' > v_{3} \Rightarrow w' \in \{ v_{4},v \} .\]
and we compute 
$F_{v_3,1}:=F(v_{3}) \cap F(v_{4})$ and $F_{v_3,2}:=F(v_{3}) \cap F(v)$
in order to compute $F(v_3) \cap \bigcup_{w' < v_3}{F(w')}$.
First, we have
\[ F_{v_3,2}=F(v_{3}) \cap F(v) = \{ F_{m+1},F_{m+2},F_{m+3} \} \]
and thus it is a facet of $F(v_3)$.
Next, we have
\[ F_{v_3,1}=F(v_{3}) \cap F(v_{4}) = \{ F_{m+1},F_{m+2},\dots \} .\]
If $F_{v_3,1}$ is not a facet of $F(v_3)$, it should be contained in at least 2 facets of $F(v_3)$.
It follows that $F_{v_3,1}$ should be contained in at least two of $F(w_1) \cap F(v_3)$, \dots , $F(w_{l-4}) \cap F(v_3)$,  $F(v_{3}) \cap F(v)$.
Therefore $F_{v_3,1}$ should be contained in $F(v_3) \cap \bigcup_{w' < v_3}{F(w')}$.
In this case, we have
\[ F(v_3) \cap \bigcup_{w' < v_3}{F(w')}= \partial F(v_3) \setminus F_{v_3,2} , \]
where $\partial F(v_3)$ denotes the union of all the facets of $F(v_3)$.

In the other case, we have
\[ F(v_3) \cap \bigcup_{w' < v_3}{F(w')}= \partial F(v_3) \setminus F_{v_3,2} \]
or
\[ F(v_3) \cap \bigcup_{w' < v_3}{F(w')}= \partial F(v_3) \setminus  (F_{v_3,1} \cup F_{v_3,2}) .\]
Thus we conclude that $F(v_3) \cap \bigcup_{w' < v_3}{F(w')}$ is the beginning of a shelling of $F(v_3)$ 
by using Lemma \ref{lem:line}.
In a similar way it can
also be proved that $F(v_4) \cap \bigcup_{w' < v_4}{F(w')}$ and  $F(v) \cap \bigcup_{w' < v}{F(w')}$ 
are the beginning of shellings of $F(v_4)$, $F(v)$, respectively.
Using the original definition of the shelling property, we can conclude that $G(P)$ satisfies the shelling property.
This completes the proof.
\end{proof}

We note here that the unique sink of $tr(G(P))$ is also simple, and we can apply the truncation operation repeatedly.

\subsection{Pyramid operation}
Next, we define the pyramid operation to generate a new $X$-type graph of a ($d+1$)-polytope with ($n+1$) vertices
from that of a $d$-polytope with $n$ vertices.
\begin{defn}(Pyramid of polytopes)\label{df:pyramid}
Given a $d$-polytope $P$ in $\Re^d$,
its {\it pyramid polytope} $py(P, v)$
is a $(d+1)$-polytope in $\Re^{d+1}$
which is the convex hull of $P \times \{ 0 \}$
and a point $v \in \Re^{d+1}$ not on the $d$-dimensional subspace
containing $P$.
A canonical choice is to set $v_{d+1} = 1$,
see Figure \ref{fig:pyramid}.
\end{defn}

\begin{figure}[ht]
\begin{center}
\includegraphics[scale=0.40]{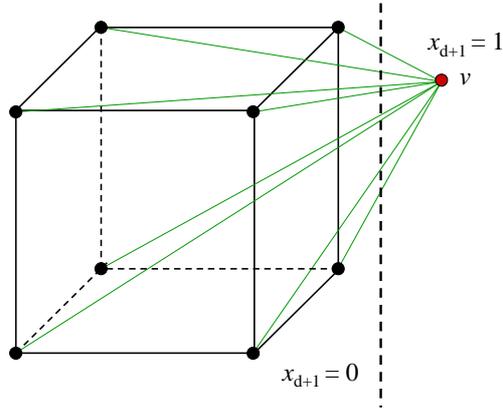}
\end{center}
\caption{A pyramid polytope of a $d$-polytope}
\label{fig:pyramid}
\end{figure}

\begin{defn}(Pyramid of polytopal digraphs)
Given a polytopal digraph $G(P)$ of a $d$-polytope $P$,
the {\it pyramid} $py(G(P),v)$ of $G(P)$ is a polytopal digraph of $(d+1)$-polytope
$py(P,v)$ whose new edges are directed as follows.
\[(v_1,v),(v_2,v),\dots ,(v_n,v), \]
for vertices $v_1,\dots ,v_n$ of $P$.
\end{defn}

\begin{lem}
Let $F_{1},\dots ,F_{m}$ be the facets of $P$.
The facets of $py(P,v)$ can be listed as follows.
\[ F_{1} \Cup \{ v \} ,\dots , F_{m} \Cup \{ v \} , P. \]
\end{lem}

We prove the following theorem.
\begin{thm} \label{prop:2ndConstruction}
Let $P$ be a $d$-polytope.
If $G(P)$ is an $X$-type graph, $py(G(P),v)$ is also an $X$-type graph.
\end{thm}
\begin{proof}
In this proof, we denote the combinatorial polar of $P$ by $P^{*}$ and that of $py(P,v)$ by $py(P,v)^{*}$.
In addition, we use the notation $F(w)$ to describe the facet of $P^*$ corresponding to 
a vertex $w$ of $P$ and the notation $F'(w')$ to describe a vertex $w'$ of $py(P,v)^{*}$.

Consider an $i$-face $H$ of $py(P)$ for $i \leq d+1$.
If $H$ is a face of $P$, the subgraph $py(G(P),v)_{H}$ induced by $H$ clearly satisfies acyclicity, USO property
and Holt-Klee property.
if $H$ is a face of $F_{j}$ for $j \leq m$, the subgraph $py(G(P),v)_{H \setminus \{ v \} }$ induced by 
the $(i-1)$-face $H \setminus \{ v \}$
satisfies acyclicity and the USO property. 
There is no edge directed from $v$ to $H \setminus \{ v \}$ and thus
$py(G(P),v)_{H}$ satisfies acyclicity.
Similarly, we see that the unique source of $py(G(P),v)_{H \setminus \{ v \} }$ is also the
unique source $s$ of $py(G(P),v)_{H}$
and $v$ is the unique sink of $py(G(P),v)_{H}$.
In addition, we notice that $py(G(P),v)_{H}$ has $i$ disjoint paths from $s$ to $v$:
\[ s \rightarrow v, \  s \rightarrow v_{1} \rightarrow v, \ \dots \ ,s \rightarrow v_{i-1} \rightarrow v \]
where $v_{1},\dots ,v_{i-1}$ are neighbours of $v$ in $py(G(P),v)_{H \setminus \{ v \} }$.
Therefore, $py(G(P),v)$ also satisfies acyclicity, the USO property and the Holt-Klee property.

In order to prove that if $G(P)$ does not satisfy the shelling property, $py(G(P),v)$ does not also satisfy the shelling
property, we prove the contrapositive. 
Assume that $F'(w_{1}),\dots ,F'(w_{n})$ is a shelling of $py(P,v)^{*}$
for a topological sort $w_{1},\dots ,w_{n}$  of $py(G(P),v)$.
First we note that $w_{n}$ must be $v$. Immediately,
\[ F(w_{1})=F'(w_{1}) \setminus \{ v \} ,\dots , F(w_{n-1})=F'(w_{n-1}) \setminus \{ v \} \]
and 
\[ F(w_{j}) \cap \bigcup_{i=1}^{j-1}{F(w_{i})} =
\bigcup_{i=1}^{j-1}{[ (F'(w_{j}) \cap F'(w_{i})) \setminus \{ v \} ]} = \bigcup_{i=1}^{r}{(G_{i} \setminus \{ v \} ) ,}\]
for $1 \leq j < n$, where $G_{1},\dots ,G_{t},\dots ,G_{r}$ is a shelling of $F'(w_{j})$.
Since $G_{1} \setminus \{ v \} ,\dots ,G_{t} \setminus \{ v \} ,\dots , G_{r} \setminus \{ v \}$ are facets of $F(w_{j})$,
we conclude that $F(w_{1}),\dots ,F(w_{n-1})$ is a shelling of $P^{*}$.
\end{proof}

\subsection{Proof of Theorem \ref{thm:main_3}}
Let $P_0$ be a 4-polytope with $n_0$ vertices  which
has an $X$-type graph $G(P_0)$ satisfying the conditions in Theorem \ref{thm:main_3}. Consider any $d \ge 4$ and $n \ge n_0+d-4$.
We apply the truncation operation $n-n_0-d+4$ times to obtain a 4-polytope $P_1$
with $n-d+4$ vertices and $X$-type graph $G(P_1)$.
Then we apply the pyramid operation $d-4$ times to $P_1$
to obtain the required polytope $P$.
The correctness of the construction follows immediately from Theorems \ref{prop:1stConstruction}
and \ref{prop:2ndConstruction}.

We note here that we can apply the construction to $G(\Omega )$ in Section 3 by putting $v:=F_{10},v_1:=F_1,v_2:=F_2,v_3:=F_3,v_4:=F_9$,
and obtain a family of $X$-type graphs $G(\Omega_{n,d})$ of $d$-polytope with $n$ vertices for $d \geq 4, n \geq 6+d$.

%

\section{LP orientations of $d$-crosspolytope}
In this section, we investigate orientations of $d$-crosspolytopes, motivated by results of Develin~\cite{Mi06}.
Develin~\cite{Mi06} introduced a fundamental notion of {\it pair sequences} for acyclic orientations of $d$-crosspolytopes.
\begin{defn}{\rm (\cite{Mi06})}\\
\label{psdef}
Let $P_{1},P_{-1},\dots,P_{d},P_{-d}$ be vertices of $d$-crosspolytope $C_d$, where
$P_{-i}$ is the unique vertex which is not connected to $P_i$ for $1 \leq i \leq d$.
Given an acyclic orientation {\it O} of $C_d$, the corresponding {\it pair sequence} is defined as a sequence $(L_1 L_2)\dots (L_{2d-1} L_{2d})$,
where $L_1,\dots,L_{2d}$ is a labeling of the vertices of $C_d$ satisfying the following conditions:
\begin{enumerate}
\item $\{ L_1,L_2,\dots,L_{2d} \} = \{ 1,2,\dots,2d \}$, 
\item $L_1 < L_3 < \dots < L_{2d-1}$,
\item $L_1 < L_2$, $L_3 < L_4$, \dots, $L_{2d-1} < L_{2d}$,
\item every edge of $C_d$ is directed from a smaller label to a larger label.
\item the pairs $(L_1,L_2),\dots,(L_{2d-1},L_{2d})$ are pairs of labels of the antipodal pairs $(P_1,P_{-1}),\dots,(P_d,P_{-d})$.
\end{enumerate}
\end{defn}
For details, see~\cite{Mi06}. We note that given a pair sequence, we can also construct 
the (not necessarily unique) corresponding acyclic orientations.
He completely characterized LP orientations of $d$-crosspolytopes using pair sequences. 
\begin{thm}{\rm (\cite{Mi06})}\\
\label{mi06}
An acyclic orientation ${\it O}$ of $d$-crosspolytope is an LP orientation if and only if
the corresponding pair sequence $(L_1 L_2)(L_3 L_4)\dots (L_{2d-1}L_{2d})$ satisfies the following condition:
\[ \{ L_1,\dots,L_{2k} \} \neq \{ 1,\dots,2k \} \text{ for all $0 < k < d$}. \]
A pair sequence satisfying the above condition is called a {\it good pair sequence}.
\end{thm}
Vertex orderings of $C_d$ which give good pair sequences are completely in one-to-one correspondence with
 the shellings of $d$-cube~\cite[Ex 8.1 (i)]{Ziegler} as noted in~\cite{Mi06}. 
This implies that the shelling property completely characterizes LP orientations of $d$-crosspolytopes.

We can also read off whether a given acyclic orientation ${\it O}$ has a unique source or not by looking at
the corresponding pair sequence $(L_1 L_2)(L_3 L_4)\dots (L_{2d-1}L_{2d})$:
${\it O}$ has a unique source if and only if $(L_1,L_2) \neq (1,2)$.
Similarly, ${\it O}$ has a unique sink if and only if $(L_{2d-1},L_{2d}) \neq (2d-1,2d)$.
We can easily see that every acyclic USO of $d$-crosspolytope satisfies the Holt-Klee property.
Using these characterizations, we can count the number of pair sequences that give acyclic USOs that satisfy the Holt-Klee property
but do not satisfy the shelling property.
\begin{cor}
Let $a_d$ be the number of of good pair sequences of
$d$-crosspolytope, and let $b_d$ be the number of 
pair sequences arising from
acyclic USOs satisfying the Holt-Klee property.
Then the following holds:
\begin{equation}
\label{ad}
 a_d = (2d-1)!!-\sum_{k=1}^{d-1}{(2d-2k-1)!!a_{k}}
\end{equation}
\begin{equation}
\label{bd}
 b_d = \{ (2d-3)^2+1\} (2d-5)!! 
\end{equation}
for $d \geq 2$, where $(-1)!!:=1$ and $a_1 = 1$.
\end{cor}
\begin{proof}
A recurrence relation for $a_d$ is given in a preprint version of \cite{Mi06}, 
but it is slightly incorrect. We follow the same approach here
correcting the error.
We use the equivalence of good pair sequences
and LP-orientations given in Theorem \ref{mi06}. 
First, the number of all pair sequences 
not necessarily satisfying condition 2 of Definition \ref{psdef} is
\[  \begin{pmatrix} 2d \\ 2 \end{pmatrix} \cdot \begin{pmatrix} 2d-2 \\ 2 \end{pmatrix} \cdot \cdots \cdot \begin{pmatrix} 2 \\ 2 \end{pmatrix}. \]
By sorting the pairs in the sequences according to the values of the first elements of the pairs, we obtain pair sequences satisfying the condition 2.
Since $d!$ sequences result in the same sequences, the number of all pair sequence is
\[  \begin{pmatrix} 2d \\ 2 \end{pmatrix} \cdot \begin{pmatrix} 2d-2 \\ 2 \end{pmatrix} \cdot \cdots \cdot \begin{pmatrix} 2 \\ 2 \end{pmatrix}/d!
=(2d-1)(2d-3)\cdots 1 = (2d-1)!!. \]
We have to subtract the number of bad pair sequences. We count the number of bad pair sequences that violate at first the condition of good pair sequences
in $k$-th pair. In this case, the sequence formed by the first $k$ pairs is a good pair sequence and the remaining part can be any pair sequences.
Therefore the number of bad pair sequences that violate at first the condition of good pair sequences in $k$-th pair is $a_k(2d-2k-1)!!$. 
Finally, we have
\[ a_d = (2d-1)!!-\sum_{k=1}^{d-1}{(2d-2k-1)!!a_{k}}.\]
Next we consider $b_d$.
If {\it O} has multiple sources, the corresponding pair sequence starts with
the pair $(12)$.
The number of sequences for this case is $(2d-3)!!$.
The number for the case where there are multiple sinks is also $(2d-3)!!$.
Since we are double-counting the case where {\it O} has multiple sources and multiple sinks, we have
\[ b_d = (2d-1)!!-2(2d-3)!!+(2d-5)!! = \{ (2d-3)^2+1\} (2d-5)!!.\]
\end{proof}

Using the recurrence relation for $a_d$, we obtain $a_d < (2d-3)(2d-3)!!$ for $d \geq 4$ (see appendix).
Therefore, the shelling property excludes $b_d-a_d > (2d-5)!!$ pair sequences from all pair sequences
that define acyclic USOs satisfying the Holt-Klee property.
Finally, we note that the ratio of the number of orientations satisfying the shelling property and that satisfying the other 3 properties 
is $\frac{a_d}{b_d} < 1 - \frac{1}{(2d-3)^2+1}.$
\section{Concluding remarks}
In this paper we introduce a new definition of the shelling property and give new 
infinite families of polytopes with $X$-type graphs, namely
polytopal digraphs which are acyclic USOs satisfying the Holt-Klee property,
but not the shelling property.
We achieved this with two operations to construct new $X$-type graphs from an $X$-type graph:
truncation and pyramid.
The pyramid operation works in any dimension but the truncation operation is restricted
to dimension 4, and to polytopes whose unique sink is simple.
It would be of interest to remove these conditions.
Those results shows abundance of $X$-type graphs and then importance of the shelling property.
We also see importance of the shelling property from a direct consequence of the result of Develin~\cite{Mi06}:
there are some polytopes whose LP orientations can completely be characterized by the shelling property but
cannot be done by the other 3 properties.


\section*{Appendix}
\begin{prop}
\begin{equation}
\label{upper}
 a_d < (2d-3)(2d-3)!!
\end{equation}
\begin{equation}
\label{lower}
 a_d > (2d-4)(2d-3)!! 
\end{equation}
for $d \geq 4$.
\end{prop}
\begin{proof}
We proceed by induction on $d$ proving (\ref{upper}) and (\ref{lower})
simultaneously.
Since $a_1=1,a_2=2,a_3=10,a_4=74,a_5=706$, the relation holds for $a_4,a_5$.
To prove (\ref{upper}), we use (\ref{ad}) and take only the first two and
last two terms of the summation. We see that for $d \ge 5$
\begin{align*}
a_{d+1}
&\le (2d+1)!!-(2d-1)!!a_1-(2d-3)!!a_2-a_d-3a_{d-1}\\
&< \{(2d+1)(2d-1)-(2d-1)-2-(2d-4)\}(2d-3)!!-3(2d-6)(2d-5)!!
~~(using~(\ref{lower}))\\
&< (4d^2 -4d+1)(2d-3)!!\\
&= (2d-1)(2d-1)!!.
\end{align*}
For (\ref{lower}), we first observe that for $4 \le k \le d-1$
\[
(2k-3)!! \le (2d-5)(2d-7)...(2d-2k+3)3
\]
and so
\begin{equation}
\label{bound}
(2d-2k+1)!!(2k-3)!! \le 3(2d-5)!!
\end{equation}
Now expanding (\ref{ad}) for $d \ge 5$
we have:
\begin{align*}
a_{d+1} 
&= (2d+1)!!-(2d-1)!!a_1-(2d-3)!!a_2-(2d-5)!!a_3 
-a_d-\sum_{k=4}^{d-1}{(2d-2k+1)!!a_k}\\
&> 2d(2d-1)!!-(2d-1)(2d-3)!!-10(2d-5)!!
-\sum_{k=4}^{d-1}{(2d-2k+1)!!(2k-3)(2k-3)!!}
~~(using~(\ref{upper}))\\
&\ge (2d-1)(2d-1)!!-\{10+3\sum_{k=4}^{d-1}{(2k-3)}\}(2d-5)!!
~~(using~(\ref{bound}))\\
&= (2d-1)(2d-1)!!-\{10+3d(d-4)\}(2d-5)!!\\
&> (2d-2)(2d-1)!!.
\end{align*}
This completes the proof.
\end{proof}

%
%

\end{document}